\documentclass[sigconf]{acmart}
\usepackage[utf8]{inputenc}
\usepackage{verbatim} 
\usepackage{graphicx} 
\usepackage{mathtools} 
\usepackage{bbm,dsfont} 
\usepackage{mathrsfs} 
\numberwithin{equation}{section} 
\usepackage{tikz-cd,tkz-graph,tkz-arith,tkz-berge} 
\usepackage{tkz-graph,tkz-arith,tkz-berge} 
\usetikzlibrary{patterns,calc}
\usepackage{float} 
\usepackage[labelfont=bf,labelsep=space]{caption} 
\usepackage{empheq} 

\usepackage{subcaption}

\usepackage{color}
\definecolor{red}{rgb}{.7,0,0}
\definecolor{blue}{rgb}{0,0,1}

\def\mcC{\mathcal{C}}

\def\mcH{\mathcal{H}}
\def\mcI{\mathcal{I}}

\def\mcL{\mathcal{L}}
\def\mcM{\mathcal{M}}
\def\mcN{\mathcal{N}}
\def\mcP{\mathcal{P}}

\def\mcS{\mathcal{S}}

\def\mcZ{\mathcal{Z}}

\def\bbR{\mathbb{R}}

\def\bbC{\mathbb{C}}

\def\bbS{\mathbb{S}}
\usepackage{ifpdf}
\ifpdf
\hypersetup{
 pdftoolbar=true,
 pdfmenubar=true,
 pdfstartview={FitV},
 pdftitle={Plantinga-Vegter algorithm takes average polynomial time},
 pdfauthor={Alperen Ali Erg\"{u}r, Felipe Cucker, Josu\'{e} Tonelli-Cueto},
 pdfsubject={computational geometry, numerical algorithms, complexity},
 pdfkeywords={real algebraic geometry,average complexity},
 pdfcreator={overleaf}, 
 pdfproducer={overleaf},
 linkcolor=red     
 citecolor=green   
 urlcolor=cyan     
 filecolor=magenta 
}
\fi
\usepackage{bookmark}

\makeatletter
\def\th@plain{%
  \thm@notefont{}
  \slshape 
}
\def\th@definition{%
  \thm@notefont{}
  \normalfont 
}
\makeatother
\theoremstyle{acmplain}
\newtheorem{lem}{Lemma}[section]
\newtheorem{prop}[lem]{Proposition}
\newtheorem{theo}[lem]{Theorem}
\newtheorem{cor}[lem]{Corollary}

\theoremstyle{acmdefinition}
\newtheorem{defi}[lem]{Definition}
\theoremstyle{remark}

\newtheorem{remark}[lem]{Remark}
\newcommand{\eproof}{\hfill\qed}

\def\Hd{\mcH_{n,d}}
\def\Pd{\mcP_{n,d}}
\def\Pdd{\mcP_{n,d-1}^n}
\def\hm{^{\mathsf{h}}}

\def\kappaff{\kappa_{\sf aff}}

\def\Oh{\mathcal{O}}

\def\Tg{\mathrm{T}}
\def\diff{\partial}
\def\Diff{\partial}
\def\pv{$\mathsf{PV}$ }
\def\vol{\mathrm{vol}\,}
\newcommand{\eps}{\varepsilon}
\newcommand{\R}{\mathbb{R}}
\newcommand{\E}{\mathbb{E}}
\newcommand{\PP}{\mathbb{P}}

\DeclarePairedDelimiter\abs{\lvert}{\rvert}
\DeclarePairedDelimiter\norm{\lVert}{\rVert}

\usepackage{multicol}
\usepackage[ruled,algosection,vlined]{algorithm2e}

\makeatletter                   
\def\mdseries@tt{m}             
\makeatother                    
\usepackage[plain]{fancyref}
\usepackage[draft=true]{minted} 
\usepackage{color}
\usepackage{hyperref}           
\hypersetup{
    colorlinks=true,
    linkcolor=blue,
    filecolor=red,      
    urlcolor=magenta,
    breaklinks=true,            
}
\usepackage{breakurl}           

\renewcommand\footnotetextcopyrightpermission[1]{} 
\acmConference[]{}{}{}



\begin{document}

\title{Plantinga-Vegter algorithm takes average polynomial time}
\author{Felipe Cucker}
\orcid{0000-0002-4569-3248}
\affiliation{
\institution{City University of Hong Kong}
\department{Dept. of Mathematics}
\streetaddress{Tat Chee Avenue}
\city{Kowloon Tong}
\country{Hong Kong}}
\email{macucker@cityu.edu.hk}
\author{Alperen A. Ergür}
\orcid{0000-0002-2340-6551}
\affiliation{
\institution{Technische Universit\"at Berlin}
\department{Institut f\"ur Mathematik}
\streetaddress{Str. des 17. Juni 136}
\city{Berlin}
\postcode{10623}
\country{Germany}}
\email{erguer@math.tu-berlin.de}
\author{Josue Tonelli-Cueto}
\orcid{0000-0002-2904-1215}
\affiliation{
\institution{Technische Universit\"at Berlin}
\department{Institut f\"ur Mathematik}
\streetaddress{Str. des 17. Juni 136}
\city{Berlin}
\postcode{10623}
\country{Germany}}
\email{ton-cue@math.tu-berlin.de}
\begin{abstract}
We exhibit a condition-based analysis of the adaptive subdivision algorithm due to 
Plantinga and Vegter. The first complexity analysis of
the \pv~Algorithm is due to Burr, Gao and Tsigaridas 
who proved a $\Oh\big(2^{\tau d^{4}\log d}\big)$ 
worst-case cost 
bound for degree $d$ plane curves with maximum 
coefficient bit-size~$\tau$. This exponential bound, 
it was observed, is in stark contrast with the 
good performance of the algorithm in practice. 
More in line with this performance, we show that, 
with respect to a 
broad family of measures, the expected time 
complexity of the \pv~Algorithm is bounded by $O(d^7)$ for real, degree 
$d$, plane curves. We also exhibit a smoothed
analysis of the \pv~Algorithm that yields similar 
complexity estimates. To obtain these results 
we combine robust probabilistic 
techniques coming from geometric functional analysis 
with condition numbers and the continuous 
amortization paradigm introduced by Burr, Krahmer 
and Yap. We hope this will motivate a 
fruitful exchange of ideas between the different approaches to numerical computation.
\end{abstract}

\maketitle
\begin{acks}
We thank Michael Burr for useful discussions. This work was supported by the \grantsponsor{}{Einstein Foundation Berlin}{https://www.einsteinfoundation.de/en/}. FC was partially supported by a GRF grant from the Research Grants Council of the Hong Kong SAR (project number CityU 11202017).
\end{acks}
\section{Introduction}

In 2004 Plantinga and Vegter proposed an algorithm for computing a regularly
isotopic piecewise linear approximation of a curve or
surface~\cite{plantingavegter2004}. Their algorithm relied on a subdivision
method enhanced with interval arithmetic to certificate the procedure
(i.e., ensure its correctness) and in Section~7 of their paper they provided
some examples with their approximations and the record of how many cubes
(squares in the case of plane curves) were in the description of these approximations.
This number of cubes appears to be proportional to, and dominate, the cost
of the computation. The paper however, contained no complexity analysis
and not even a formal setting fixing either the kind of functions implicitly defining
the considered curves and surfaces or the arithmetic used.

An article doing so was published in 2017 by Burr, Gao and
Tsigaridas~\cite{burr2017}. The functions this article deals with are polynomials
with integer coefficients and smooth zero set. Consistently with this choice of data,
the arithmetic is infinite precision. The main result in the paper is a worst-case
complexity
analysis for the number of cubes in the description of the approximation which,
as we just mentioned, dominates the cost of the computation. The bounds proved for this
quantity are shown to be optimal. Yet, these bounds
are exponential (both in the degree of the input polynomial and in its
logarithmic height), a fact that motivates the following comment at the end
of the paper 
\begin{quote}
  Even though our bounds are optimal, in practice, these are quite
  pessimistic [\dots]
\end{quote}
The authors further observe that, following from their Proposition~5.2 
(see Theorem~\ref{theo:analysis2} below) an instance-based analysis of the algorithm (i.e., one yielding a cost that 
depends on the input at hand) could be derived
from the evaluation of a certain integral. And they conclude their paper by writing  
\begin{quote}
  Since the complexity of the algorithm can be exponential in the inputs [size],
  the integral must be described in terms of additional geometric and intrinsic
  parameters. 
\end{quote}
A number of features in this state of affairs suggest that a condition-based
approach to the analysis of our quantity of interest 
could be useful. To begin with, the fact that a condition number is a
perfect fit for the notion of an ``additional geometric and intrinsic parameter.''
To which we may add the fact that the obvious set of ill-posed inputs, the set
of polynomials having a non-smooth zero set, is precisely the set of data which
are not allowed as inputs in~\cite{burr2017}. Of course, such a condition-based
analysis would drop the assumption of integer coefficients and replace it by
that of real coefficients but this is a common practice for numerical algorithms
and, as we will see, it pays off in our case as it yields small (i.e., polynomial)
average complexity bounds for a large class of probability measures. 

Although our approach follows the condition-based 
ideas of, e.g.,~\cite{Demmel88,Renegar95,Condition,CKMW1,CKMW2}, the 
complexity analysis in this paper would have been 
impossible without the~\emph{continuous amortization}
technique developed in the exact numerical  
context~\cite{burr2009,burr2016}. We hope that 
this merging of techniques will start a fruitful
exchange of ideas between different approaches to
continuous computation.

\subsection{Notation}

Throughout the paper, we will assume some familiarity with the basics
of differential geometry and with the sphere 
$\bbS^n$ as a Riemannian manifold. For scalar smooth maps
$f:\bbR^m\rightarrow \bbR$, we will write the tangent map at
$x\in\bbR^m$ as $\diff_xf:\bbR^m\rightarrow \bbR$ when we want to
emphasize it as a linear map and as $\diff f:\bbR^m\rightarrow
\bbR^m$, $x\mapsto \diff f(x)$, when we want to emphasize it as a
smooth function. For general smooth maps $F:\mcM\rightarrow \mcN$, we
will just write $\Diff_xF:\Tg_x\mcM\rightarrow \Tg_x\mcN$ as the
tangent map.

In what follows, $\Pd$ will denote the set of real polynomials in $n$
variables with degree at most $d$, 
$\Hd$ the set of homogeneous real
polynomials in $n+1$ variables of degree $d$, and 
$\|~\|$ and 
$\langle\,~,~\rangle$ will denote the usual norm
and inner product in $\bbR^m$ as well as
the Weyl norm and inner product in $\Pd^m$ and $\Hd^m$.
Given a polynomial $f\in \Pd$, $f\hm\in\Hd$ will be its
homogenization and $\diff f$ the polynomial map given by its
partial derivatives. For details about the concrete definition of each
of these notions, see Section~\ref{sec:geomfram}. Additionally, $V_\bbR(f)$ and $V_\bbC(f)$ will be, respectively, the real and complex zero sets of $f$.

We will denote by
$\mcI_n$ the set of $n$-cubes of $\bbR^n$ and, for a given $J\in\mcI_n$,
$m(J)$ will be its middle point, $w(J)$ its width, and
$\vol J=w(J)^n$ its volume.

Also, $\PP(A)$ will denote the
probability of the event $A$, $\E_{x\in K}g(x)$ the expectation of
$g(x)$ when $x$ is sampled uniformly from $K$ and $\E_{y}g(x)$ the
expectation of $g(y)$ with respect to a previously specified probability
distribution of $y$. 

Regarding complexity parameters, $n$ will be the
number of variables, $d$ the degree bound, and $N=\binom{n+d}{n}$ the
dimension of $\Pd$. 

Finally, $\ln$ will denote the natural 
logarithm and $\log$ the logarithm in base $2$.

\subsection{Outline}

In Section~\ref{sec:pv}, we discuss the \pv~Algorithm and the $n$-dimensional
generalization of its subdivision method that we will analyze. In Section~\ref{sec:main}, we state the main complexity results of this paper. In
Section~\ref{sec:geomfram}, we present the geometric framework of polynomials we
will work with. Following a common practice in condition-based analysis
we use homogenization to get many of our results. In Section~\ref{sec:condition}, we introduce
the condition number along with some of its main properties.
In Section~\ref{sec:complexity}, we present the existing results of complexity
of the subdivision method of the \pv~Algorithm based on local size
bound functions from~\cite{burr2017} and we relate them to the local
condition number. In Section~\ref{sec:probability}, we rely on the bounds
for the condition number obtained in Section~5 to derive average and smoothed complexity
bounds under (quite) general randomness assumptions.

\section[The PV Algorithm]{The \pv Algorithm}\label{sec:pv}

Given a real smooth hypersurface in $\bbR^n$ described implicitly by a map $f:\bbR^n\rightarrow \bbR$ and a region $[-a,a]^n$, the \pv~Algorithm constructs a piecewise-linear approximation of the intersection of its zero set $V_\bbR(f)$ with $[-a,a]^n$ isotopic to this intersection inside $[-a,a]^n$.

Let $\mcI_m$ be the set of $m$-cubes of $\bbR^m$. Recall that an \emph{interval 
approximation} of a function $F:\bbR^{m}\rightarrow \bbR^{m'}$ is a map $\square[F]:\mcI_m\rightarrow \mcI_{m'}$ such that for all $J\in \mcI_m$, $F(J)\subseteq \square[F](J)$ (c.f.~\cite{ratschek1984}). We notice that if we see $J$ as error bounds for the midpoint $m(J)$, then $\square[F](J)$ is nothing more than error bounds for $F(m(J))$.

Assume that we have interval approximations of both $f$ and its tangent map $\diff f$ or, more generally, of $hf$ and $h'\diff f$ for some positive maps $h,h':\bbR^n\rightarrow (0,\infty)$. The \pv Algorithm on $[-a,a]^n$ will subdivide this region 
into smaller and smaller cubes until the condition
\[C_f(J)\text{: either }0\notin \square[hf](I)\text{ or }0\notin \langle \square[h'\diff f](J),\square[h'\diff f](J)\rangle\]
is satisfied in each of the $n$-cubes $J$ of the obtained subdivision of $[-a,a]^n$. In Section~\ref{sec:geomfram}, we will be more precise on the assumptions on our interval approximations and the functions $h$ and $h'$ that we will use.

\begin{algorithm}
\DontPrintSemicolon
\SetKwInput{input}{Input}
\SetKwInput{output}{Output}
\caption{Subdivision routine of \pv Algorithm}\label{alg:PVAlgorithm}
\input{$a \in (0,\infty)$ and $f:\bbR^n\rightarrow \R$\\
with interval approximations $\square[hf]$ and $\square[h'\diff f]$
}
\hrulefill

Starting with the trivial subdivision $\mcS:=\{[-a,a]^n\}$, repeatedly subdivide each $J\in\mcS$ into $2^n$ cubes until the condition $C_f(J)$ holds for all $J\in\mcS$\;
\hrulefill

\output{Subdivision $\mcS\subseteq \mcI_n$ of $[-a,a]^n$\newline such that for all $J\in \mcS$, $C_f(J)$ is true}
\end{algorithm}

The procedure in Algorithm~\ref{alg:PVAlgorithm} is only the subdivision routine 
of the \pv~Algorithm but it dominates its complexity
in the sense that the remaining computations do 
not add to the final cost estimates in 
Landau notation. Moreover, these additional 
computations have been implemented only for 
$n\leq 3$. So, proceeding as in~\cite{burr2017}, 
we will only analyze the complexity of the 
subdivision routine, keeping track of the 
dependency on $n$. Also as in~\cite{burr2017}, 
our complexity analysis will not deal with the
precision needed for the algorithm.

\section{Main result}\label{sec:main}

In this section, we outline without proofs the 
main results of this paper. In the first part, 
we describe our randomness assumptions for polynomials. In the second one, we give precise statements for our bounds on the average and smoothed complexity of the \pv Algorithm.

\subsection{Randomness Model}
Most of the literature on random multivariate polynomials considers polynomials with Gaussian independent coefficients and relies on techniques 
that 
are only useful for Gaussian measures. We will 
instead consider a general family of measures 
relying on robust techniques coming from 
geometric functional analysis. Let us recall 
some basic definitions. 

\begin{enumerate}
\item[P1] A random variable $X$ is called \emph{centered} if $\E X =0$. 
\item[P2] A random variable $X$ is called \emph{subgaussian} if there exist a $K$ such that for all $p \geq 1$,
\[ (\E \abs{X}^p)^{\frac{1}{p}} \leq K \sqrt{p}. \]
The smallest such $K$ is called the $\Psi_2$-norm 
of $X$. 
\item[P3] A random variable $X$ satisfies the 
\emph{anti-concentration property with constant $\rho$} if
\[ \max\left\{ \PP \left( \abs{X-u} \leq \eps \right) \mid u \in \bbR \right\}\leq \rho \eps .\]
\end{enumerate}

The subgaussian property (P2) has other equivalent formulations. We refer the interested reader 
to~\cite{V}.

\begin{defi}
A \emph{dobro random polynomial} $f\in \Hd$ with parameters $K$ and $\rho$ is a polynomial
\begin{equation} \label{randomdef}
f:=\sum_{|\alpha|=d}\binom{d}{\alpha}^{1/2}\;c_\alpha X^\alpha
\end{equation}  
such that the $c_\alpha$ are independent 
centered subgaussian random variables with $\Psi_2$-norm $\leq K$ 
and anti-concentration property with constant $\rho$.
A \emph{dobro random polynomial} $f\in \Pd$ is a polynomial $f$ such that its homogenization $f\hm$ is so. 
\end{defi}

Some dobro random polynomials of interest are the following three.

\begin{enumerate}
    \item[N] A \emph{KSS random polynomial} is a dobro random polynomial such that each $c_\alpha$ in \eqref{randomdef} is Gaussian with unit variance. For this model we have $K\rho = 1/\sqrt{2\pi}$.
    \item[U] A \emph{Weyl random polynomial} is a dobro random polynomial such that each $c_\alpha$ in \eqref{randomdef} have uniform distribution in $[-1,1]$. For this model we have 
    $K\rho\leq 1$.
    \item[E] A \emph{$p$-random polynomial} is a dobro random polynomial whose coefficients are independent identically distributed random variables with the density function $g(t)=c_p e^{-|t|^{p}}$ with $c_p$ being the appropriate constant and $p \geq 2$.
\end{enumerate}

\begin{remark}\label{rem:unif}
When we are interested in integer polynomials, dobro
random polynomials may seem inadequate. One may be 
inclined to consider random polynomials $f\in\Pd$ 
such that $c_\alpha$ is a random integer in the
interval $[-2^{\tau},2^\tau]$, i.e., $c_\alpha$ is 
a random integer of bit-size at most $\tau$. As
$\tau\to\infty$ and after we normalize the 
coefficients dividing by $2^\tau$, this random 
model converges to that of Weyl random polynomials. 

To have a more satisfactory understanding of random integer polynomials, one has to consider 
random variables without a continuous density function. The techniques used in this note are already extended to include such random variables in the case of random matrices ~\cite{RV, V}. 
We hope to pursue this delicate case in a more general setting (including complete intersections) in future work. 
\end{remark}

\subsection{Average and Smoothed Complexity}

The following two theorems give bounds for, respectively, the average and smoothed complexity of Algorithm~\ref{alg:PVAlgorithm}. In both 
of them $c_1$ and $c_2$ are, respectively, the universal constants 
in Theorems~\ref{thm:probtool1} and~\ref{thm:probtool2}. 

\begin{theorem}\label{expected}
Let $f\in\Pd$, $\sigma>0$, and $g\in\Pd$ a dobro random polynomial with parameters $K$ and $\rho$.
The expected number of $n$-cubes in the final subdivision of Algorithm~\ref{alg:PVAlgorithm} on input $(f,a)$ is at most
\[  d^{\frac{n^2+3n}{2}}\max\{1,a^n\} 2^{\frac{n^2+16 n \log(n)}{2}}   (c_1 c_2 K\rho)^{n+1} \]
if the interval approximations satisfy~\eqref{eq:intcond1} 
and~\eqref{eq:intcond2} and
\[ d^{\frac{n^2+5n}{2}}\max\{1,a^n\} 2^{\frac{7n^2+9 n \log(n)}{2}}  (c_1 c_2 K\rho)^{n+1}  \]
if they satisfy the hypothesis of~\cite{burr2017}.
\end{theorem}

\begin{theorem}\label{smoothed}
Let $f\in\Pd$, $\sigma>0$, and $g\in\Pd$ a dobro random polynomial with parameters $K$ and $\rho$ . Then the expected number of $n$-cubes of the final subdivision of Algorithm~\ref{alg:PVAlgorithm} for input $(q_\sigma,a)$ where $q_\sigma=f+\sigma\|f\|g$ is at most
\[ d^{\frac{n^2+3n}{2}} \max\{1,a^n\}  2^{\frac{n^2+16 n \log(n)}{2}}     \left( c_1 c_2  K\rho \right)^{n+1} \left(1+\frac{1}{\sigma}\right)^{n+1} \]
if the interval approximations satisfy~\eqref{eq:intcond1} 
and~\eqref{eq:intcond2} and
\[  d^{\frac{n^2+5n}{2}}\max\{1,a^n\} 2^{\frac{7n^2+9 n \log(n)}{2}}   
 (c_1 c_2 K\rho)^{n+1} \left(1+\frac{1}{\sigma}\right)^{n+1} \]
if they satisfy the hypothesis of~\cite{burr2017}.
\end{theorem}

If we compare the results above with the worst-case bound of~\cite[Theorem~4.3]{burr2017}, which is 
$$2^{\Oh(nd^{n+1}(n\tau+nd\log{(nd)}+9n+d)\log{a})}$$
with $\tau$ being the largest bit-size of the coefficients of $f$, we can see that our 
probabilistic bounds 
are exponentially better: they may provide 
an explanation of the efficiency of the 
\pv~Algorithm in practice.

We note, however, that the bound
in~\cite{burr2017} and our bounds cannot be 
directly compared. Not only because the former is 
worst-case and the latter average-case (or smoothed) but because 
of the different underlying complexity 
settings: the bound in~\cite{burr2017} applies to 
integer data, ours to real data. A first approach to bridge this difference relies on the approximation 
of distributions described in Remark~\ref{rem:unif}. 
But, as mentioned there, this approach does 
not give completely 
satisfactory results. A more detailed study of how
the \pv~Algorithm behaves on random integer polynomials is desirable.

\section{Geometric framework}\label{sec:geomfram}

There is an extensive literature on norms of
polynomials and their relation to norms of gradients
in $\Hd$. We can use homogenization to 
carry these results from $\Hd$ to $\Pd$. 

To be more precise, let $\phi:\bbR^n\rightarrow \bbS^n$, given by 
\[ x\mapsto \begin{pmatrix}1 & x^T\end{pmatrix}^T /\sqrt{1+\|x\|^2} .\] 
This gives a diffeomorphism between $\bbR^n$ and the upper half of $\bbS^n$, and we have 
\begin{equation}\label{eq:comphomphi}
f\hm(\phi(x))=f(x)/(1+\|x\|^2)^{d/2}.
\end{equation}
Using the chain rule, we can see that
\begin{equation}\label{eq:chainrule}
\diff_{\phi(x)}f^h \Diff_x\phi=\frac{\diff_xf}{(1+\|x\|^2)^{d/2}}
-\frac{d\cdot f(x)x^T}{(1+\|x\|^2)^{d/2+1}}
\end{equation}
where $\diff_yf\hm:\bbR^{n+1}\rightarrow \bbR$, $\Diff_x\phi:\bbR^n\rightarrow \Tg_x\bbS^n = x^{\perp}$ and $\diff_xf:\bbR^n\rightarrow \bbR$ are respectively the tangent maps of $f\hm$, $\phi$ 
and~$f$.

As it will be useful later, we note that a direct computation shows
\begin{equation}\label{eq:boundsphi}
\|\Diff_x\phi\|=1/\sqrt{1+\|x\|^2}\text{ and }\|\Diff_x\phi^{-1}\|=1+\|x\|^2. 
\end{equation}

An inner product on $\Hd$ with desirable geometric properties is known as the \emph{Weyl inner product} and it is given by
\[\langle f,g\rangle_W:=\sum_{\alpha}\binom{d}{\alpha}^{-1}f_\alpha g_\alpha\]
for $f=\sum_\alpha f_\alpha X^\alpha,\, g=\sum_\alpha g_\alpha X^\alpha\in\Hd$. This product is extended 
to $\Pd$ using $\langle f,g\rangle:=\langle f\hm,g\hm\rangle$ and to $\Pd^k$ using $\langle \mathbf{f},\mathbf{g}\rangle:=\sum_{i=1}^k\langle f_i,g_i\rangle$.

We will use $\|~\|$ to denote both the Weyl norm for polynomials and the usual Euclidean norm in $\bbR^n$. 

\subsection{Lipschitz properties}

Given $f\in \Pd$, let us consider the maps
\[\widehat{f}:x\mapsto f(x)/\left(\|f\|(1+\|x\|^2)^{(d-1)/2}\right)\]
and
\[\widehat{\diff f}:x\mapsto \diff f(x)/\left(d\|f\|(1+\|x\|^2)^{d/2-1}\right)\]
which are just "linearized" version of $f$ and 
its derivative. The intuition behind this fact 
is that for large values of $x$ a polynomial map 
of degree $d$ grows like $\|x\|^d$.

\begin{prop}\label{prop:lipschitz}
Let $\mathbf{f}\in \Pd^k$ be a polynomial map. 
Then the map 
\[
x\mapsto \tilde{\mathbf{f}}(x):=\mathbf{f}(x)/\left(\|\mathbf{f}\|(1+\|x\|^2)^{(d-1)/2}\right)
\]
is $(1+\sqrt{d})$-Lipschitz and, for all $x$, $\big\|\tilde{\mathbf{f}}(x)\big\|\leq \sqrt{1+\|x\|^2}$.
\end{prop}
\begin{proof}
For the Lipschitz property, it is enough to bound the norm of the derivative of the map by $1+\sqrt{d}$. Due to~\eqref{eq:comphomphi},
\[
\mathbf{f}(x)/\left(\|\mathbf{f}\|(1+\|x\|^2)^{(d-1)/2}\right)=\sqrt{1+\|x\|^2}\mathbf{f}\hm(\phi(x))/\|\mathbf{f}\|
\]
and so the derivative equals
\[
\frac{\mathbf{f}\hm(\phi(x))}{\|\mathbf{f}\|}
\frac{x^T}{\sqrt{1+\|x\|^2}}+\sqrt{1+\|x\|^2}\,
\frac{\Diff_{\phi(x)} \mathbf{f}}{\|\mathbf{f}\|}\Diff_x\phi.
\]
Now, $\big\|\mathbf{f}\hm(\phi(x))\big\|/
\|\mathbf{f}\|\leq 1$, by~\cite[Lemma~16.6]{Condition}, and 
\[ 
\Diff_{\phi(x)} \mathbf{f}\,\Diff_x\phi = \Diff_{\phi(x)} \left(\mathbf{f}_{|\bbS^n}\right)\Diff_x\phi 
\]
with $\left\|\Diff_{\phi(x)} \left(\mathbf{f}_{|\bbS^n}\right)\right\|\leq \sqrt{d}\|\mathbf{f}\|$, by the Exclusion Lemma~\cite[Lemma~19.22]{Condition}. 
Thus taking norms and using~\eqref{eq:boundsphi}
finishes the proof.

The claim about $\big\|\tilde{\mathbf{f}}(x)\big\|$ is just~\cite[Lemma~16.6]{Condition}.
\end{proof}

\begin{cor}\label{cor:lipschitz}
Let $f\in\Pd$. Then $\widehat{f}$ and 
$\widehat{\diff f}$ are Lipschitz with Lipschitz constants $(1+\sqrt{d})$ and $(1+\sqrt{d-1})$, respectively, and for all $x$,  $|\widehat{f}(x)|,\|\widehat{\diff f}(x)\|\leq \sqrt{1+\|x\|^2}$.
\end{cor}

\begin{proof}
The claims about $\widehat{f}$ are immediate from Proposition~\ref{prop:lipschitz}.  For the claims about $\widehat{\diff f}$, observe that $\diff f\in\Pdd$ and that, by a direct computation, $\|\diff f\|\leq d\|f\|$.
Thus Proposition~\ref{prop:lipschitz} completes the proof.
\end{proof}

\subsection[Interval approximations]{Interval approximations}

The Lipschitz properties above  (Corollary~\ref{cor:lipschitz}) ensure that the 
mapping 
\[J\mapsto \widehat{f}(m(J))+(1+\sqrt{d})\sqrt{n}w(J)\left[-1/2,1/2\right]
\]
is an interval approximation of $f\in\Pd$ and 
the mapping 
\[J\mapsto \widehat{\diff f}(m(J))+\big(1+\sqrt{d-1}\big)\sqrt{n}\,w(J)\left[-1/2,1/2\right]^n\]
one of $\diff f$. Also, letting 
\[
h(x)=\frac{1}{\|f\|(1+\|x\|)^{(d-1)/2}}\text{ and } h'(x)=\frac{1}{d\|f\|(1+\|x\|)^{d/2-1}}
\]
the mappings above become interval approximations $\square[h f]$ and $\square[h'\diff f]$ of 
$hf$ and $h'\diff f$, respectively, satisfying that, 
for each $J\in\mcI_n$,
\begin{equation}\label{eq:intcond1}
\mathrm{dist}\left(\square[hf](J),(hf)(m(J))\right)\leq \big(1+\sqrt{d}\big)\sqrt{n}\,w(J)/2
\end{equation}
and
\begin{equation}\label{eq:intcond2}
\mathrm{dist}\left(\square[h'\diff f](J),(h'\diff f)(m(J))\right)\leq \big(1+\sqrt{d-1}\big)nw(J)/2.
\end{equation}

With these conditions on our interval 
approximations, we can reformulate a weaker, but
easier to check, condition $C_f'(J)$.

\begin{theo}\label{theo:condprime}
Let $J\in\mcI_n$ and assume that our interval approximation satisfies~\eqref{eq:intcond1} and~\eqref{eq:intcond2}. If the condition
\[
\arraycolsep=0.1pt
C_f'(J)\,:\,\left\{\begin{array}{rl}
\left|\widehat{f}(m(J))\right|>&(1+\sqrt{d})\sqrt{n}w(J)\\
&\text{or }\left\|\widehat{\diff f}(m(J))\right\|>\sqrt{2}(1+\sqrt{d-1})nw(J)
\end{array}\right.
\]
is satisfied, then $C_f(J)$ is true.
\end{theo}
\begin{lem}\label{lem:innerproductbound}
Let $x\in \bbR^n$. Then for all $v,w\in B_{\|x\|/\sqrt{2}}(x)$, $\langle v,w\rangle>0$.
\end{lem}

\begin{proof}
Let $W:=\{u\in\bbR^n\mid \langle x,u\rangle\geq
\|x\|\|u\|/\sqrt{2}\}$ be the convex cone of 
those vectors $u$ whose angle $\widehat{x,u}$ 
with $x$ is at most $\pi/4$. For $v,w\in W$, $\widehat{v,w}\leq \widehat{v,x}+\widehat{x,w}\leq \pi/2$, by the triangle inequality. Thus $\cos\,\widehat{v,w}\geq 0$.

Now, $\mathrm{dist}(x,\partial W)=\min\{\|x-u\|\mid \langle x,u\rangle=\|x\|\|u\|/\sqrt{2}\}$ where the latter equals the distance of $x$ to a line having an angle $\pi/4$ with $x$, which is $\|x\|/\sqrt{2}$. Hence $B_{\|x\|/\sqrt{2}}(x)\subseteq W$ and we are done.
\end{proof}
\begin{proof}[Proof of Theorem~\ref{theo:condprime}]
When the condition on $\widehat{f}(m(J))$ is satisfied,
~\eqref{eq:intcond1} guarantees that $0\notin \square[hf](J)$. 
Whenever the condition on $\widehat{\diff f}(m(J))$ is satisfied, 
\eqref{eq:intcond2} and Lemma~\ref{lem:innerproductbound} guarantee that 
$0\notin \langle \square[h'\diff f](J),\square[h'\diff f](J)\rangle$. Hence $C_f'(J)$ implies $C_f(J)$ under the given assumptions.
\end{proof}

\begin{remark}
The interval approximations in~\cite{burr2017} are based on Taylor expansion at the midpoint, so they are different from ours. However, our complexity analysis also applies to the interval approximations 
considered in~\cite{burr2017}, see \S\ref{subsec:burrlocalsizebound} below 
for the details.
\end{remark}

\section{Condition number}\label{sec:condition}

As other numerical algorithms in computational geometry, the \pv~Algorithm has a cost which
significantly 
varies with inputs of the same size. One wants 
to explain this variation in terms of geometric properties of the input. Condition numbers allow 
for such an analysis.

\begin{defi}\cite{CKS16,BCL17}\label{def:condition}
Given $F\in\Hd$, the \emph{local condition number} of $F$ at $y\in\bbS^n$ is
\[
\kappa(F,y):=\|F\|/\sqrt{F(y)^2+\|\diff_y F_{|\Tg_y\bbS^n}\|^2/d}.
\]
Given $f\in\Pd$, the \emph{local affine condition number} of $f$ at $x\in\bbR^n$ is $\kappaff(f,x):=\kappa(f\hm,\phi(x))$.
\end{defi}

\subsection[What does kappaff measure?]{What does $\kappaff$ measure?}

The nearer the hypersurface $V_\bbR(f)$ is to having 
a singularity at $x\in\bbR^n$,  the smaller are 
the boxes drawn by the \pv Algorithm around $x$. 
A quantity controlling how close if $f$ to have a 
singularity at $x$ will therefore control the size 
of these boxes. This is precisely what 
$\kappaff(f,x)$ does. 

\begin{theo}[Condition Number Theorem]\label{cor:conditionumberthm}
Let $x\in\bbR^n$ and $\Sigma_x$ be the set of
hypersurfaces having $x$ as a singular point.  
That is, $\Sigma_x:=\{g\in\Pd\mid g(x)=0,\,
\diff_x g=0\}$.
Then for every $f\in\Pd$,
\[ 
\|f\|/\kappaff(f,x)=\mathrm{dist}
(f,\Sigma_x) 
\]
where the distance is induced by the Weyl norm 
of $\Pd$.
\end{theo}

\begin{proof}
This follows from \cite[Proposition~19.6]{Condition},
\cite[Theorem 4.4]{BCL17} and the definition of
$\kappaff$.
\end{proof}

Theorem~\ref{cor:conditionumberthm} provides a
geometric interpretation of the local condition
number, and a corresponding "intrinsic" complexity
parameter as desired by the authors 
of~\cite{burr2017,burr2018}.
The next result will be useful in the 
probabilistic analyses.

\begin{cor}\label{cor:orthogonalprojection}
Let $x\in\bbR^n$ and let
$P_x:\Pd\rightarrow \Sigma_x^\perp$
be the orthogonal projection onto the orthogonal
complement of the linear subspace $\Sigma_x$. Then
$\kappaff(f,x)=\|f\|/\|P_xf\|$.
\end{cor}

\begin{proof}
We have that $\mathrm{dist}(f,\Sigma_x)
=\|P_x f\|$ since $\Sigma_x$ is a linear subspace.
Hence Theorem~\ref{cor:conditionumberthm} finishes 
the proof.
\end{proof}

We notice that the above expression should not come 
as a surprise, since $\kappaff$ is define in a way
that the denominator is the norm of a vector 
depending linearly of $f$.

\subsection{A fundamental proposition}

The following result plays a fundamental role in 
our development. 
Despite having been used in many occasions within various proofs, it wasn't explicitly stated 
until recently.

\begin{lem}\label{prop:fundamentalproposition_hom}
Let $F\in\Hd$ and $y\in\bbS^n$. Then either
\[|F(y)|/\|F\|\geq 1/\left(\sqrt{2}\,\kappa(F,y)\right)\]
or
\[\|\diff_yF_{|\Tg_y\bbS^n}\|/\left(\sqrt{d}\,\|F\|\right)\geq 1/\left(\sqrt{2}\,\kappa(F,y)\right).\]
\end{lem}

\begin{proof}
This is~\cite[Proposition~3.6]{BCTC1}.
\end{proof}

\begin{prop}\label{prop:fundamentalproposition_aff}
Let $f\in\Pd$ and $x\in\bbR^n$. Then either
\[|\widehat{f}(x)|> \frac{1}{2\sqrt{2d}\,\kappaff(f,x)}  
\text{ or }
\left\|\widehat{\diff f}(x)\right\|> \frac{1}{2\sqrt{2d}\,\kappaff(f,x)}.\]
\end{prop}

\begin{proof}
Without loss of generality assume that $\|f\|=1$. Let $y:=\phi(x)$, $F:=f\hm$ and assume that the first inequality does not hold. Then, by~\eqref{eq:comphomphi}, $|F(y)|\leq 1/\left(2\sqrt{2d}\,\kappa(F,y)\sqrt{1+\|x\|^2}\right)$.

By~\eqref{eq:chainrule}, ~\eqref{eq:boundsphi} and Lemma~\ref{prop:fundamentalproposition_hom}, we get
\[\frac{1}{\sqrt{2}\,\kappa(F,y)}\leq \left\|\frac{\diff_xf}{(1+\|x\|^2)^{d/2}}-\frac{df(x)x^T}{(1+\|x\|^2)^{d/2+1}}\right\|\left(\frac{1+\|x\|^2}{\sqrt{d}}\right).\]

We divide by $\sqrt{d}$ and use the triangle inequality to obtain
\[\frac{1}{\sqrt{2d}\,\kappa(F,y)}\leq \frac{\|\diff_xf\|}{d(1+\|x\|^2)^{d/2-1}}+\frac{|f(x)|}{(1+\|x\|^2)^{(d-1)/2}}\frac{\|x\|}{\sqrt{1+\|x\|^2}}.\]

Using~\eqref{eq:comphomphi} and our initial 
assumption on the second term in the sum, which 
we subtract, we get the desired inequality 
since $\|x\|< \sqrt{1+\|x\|^2}$.
\end{proof}
\section{Adaptive Complexity Analysis}\label{sec:complexity}

As stated in Section~\ref{sec:pv} (and as done in~\cite{burr2017}), our complexity analysis will 
focus on the number of subdivisions steps of the subdivision routine of the \pv Algorithm (Algorithm~\ref{alg:PVAlgorithm}). That is, 
our cost measure will be the number of $n$-cubes in the final subdivision of $[-a,a]^n$. 

We note that we do not deal with the precision 
needed to run the algorithm. This issue should 
be treated in the future. 

\subsection{Local size bound framework\texorpdfstring{~\cite{burr2017}}{}}

The original analysis in~\cite{burr2017} was based on the notion of local size bound.

\begin{defi}
A \emph{local size bound} for $f$ is a function
$b_f:\bbR^n\rightarrow [0,\infty)$ such that for 
all $x\in\bbR^n$,
\[
b_f(x)\leq \inf\{\vol(J)\mid x\in J\in \mcI_n
\text{ and }C_f(J)\text{ false}\}.
\] 
\end{defi}

Arguing as in \cite[Proposition~4.1]{burr2017}, one can easily get the following general bound.

\begin{prop}\label{theo:analysis1}\cite{burr2017}
The number of $n$-cubes of the final subdivision 
of Algorithm~\ref{alg:PVAlgorithm} on input 
$(f,a)$, regardless of how the subdivision step 
is done, is at most
\begin{equation}\tag*{\qed}
(2a)^n/\inf\{b_f(x)\mid x\in [-a,a]^n\}.
\end{equation}
\end{prop}

The bound above is worst-case, it considers  
the worst $b_f(x)$ among the   
$x\in[-a,a]^n$. Continuous amortization developed 
by Burr, Krahmer and Yap~\cite{burr2009,burr2016},
provides the following refined complexity 
estimate~\cite[Proposition~5.2]{burr2017} which 
is adaptative.

\begin{theo}\label{theo:analysis2}\cite{burr2009,burr2016,burr2017}
The number of $n$-cubes of the final subdivision of Algorithm~\ref{alg:PVAlgorithm} on input $(f,a)$ is at most
\[
\max\left\{1,\int_{[-a,a]^n}\,\frac{2^n}{b_f(x)}\,
\mathrm{d} x\right\}.
\]
Moreover, the bound is finite if and only if the algorithm terminates.\eproof
\end{theo}

To effectively use Theorem~\ref{theo:analysis2} 
we need to explicit estimates for the local 
size bound.

\subsection{Construction in\texorpdfstring{~\cite{burr2017}}{Burr, Gao and Tsigarias (2017)}}\label{subsec:burrlocalsizebound}

In~\cite{burr2017}, the authors use the following function 
 \begin{equation*}
     \begin{split}
        \mcC(f,x)  :=\min \Bigg\{& \frac{2^{n-1}d/\ln{\left(1+2^{2-2n}\right)}+\sqrt{n}/2}{\mathrm{dist}(x,V_\bbC(f) )} ,  \\
& \qquad\qquad\frac{2^{2n}(d-1)/\ln{\left(1+2^{2-4n}\right)}+\sqrt{n/2}}{\mathrm{dist}( (x,x),V_\bbC(g_f))} \Bigg\}
     \end{split}
 \end{equation*}
where $g_f$ is the polynomial $\langle \diff
f(X),\diff f(Y)\rangle$, to 
construct a local size bound.

\begin{theo}\cite{burr2017}
Assume that the interval approximation satisfies the hypothesis of~\cite{burr2017}. Then
\[
  x\mapsto 1/\mcC(f,x)^n
\]
is a local size bound function for $f$.\eproof
\end{theo}

Looking at the definition of $\mcC(f,x)$ in
\cite{burr2017} one can see that $1/\mcC$ measures 
how near is $x$ of being a singular zero of $f$. 
This is similar to $1/\kappaff$ which, by
Theorem~\ref{cor:conditionumberthm}, measures how 
near is $f$ of having a singular zero at $x$. 
The following result relates these two quantities.

\begin{theo}\label{thm:boundBGTbycond}
Let $d>1$ and $f\in\Pd$. Then, for all $x\in\bbR^n$,
\[
 \mcC(f,x)\leq 2^{3n}d^2\kappaff(f,x).
\]
\end{theo}

\begin{proof}
Note that Corollary~\ref{cor:lipschitz} holds 
over the complex numbers as well. Due to this 
and the fact that $V_\bbC(f)=V_\bbC(\widehat{f})$,
we have that
\[
\big|\widehat{f}(x)\big|\leq (1+\sqrt{d})\,\mathrm{dist}(x,\mcZ_\bbC(f)).
\]

Now, if $\sqrt{2}(1+\sqrt{d-1})\,\mathrm{dist}
((y_1,y_2),(x,x))<\|\widehat{\diff f}(x)\|$, 
then $\sqrt{2}(1+\sqrt{d-1})\|y_i-x\|<\|
\widehat{\diff f}(x)\|$. Thus, by
Corollary~\ref{cor:lipschitz},
$\sqrt{2}\|\widehat{\diff f}(y_i)-
\widehat{\diff f}(x)\|<\|\widehat{\diff f}(x)\|$ 
and so, by Lemma~\ref{lem:innerproductbound}, 
$0\neq \langle \widehat{\diff f}(y_1), 
\widehat{\diff f}(y_2)\rangle$. Hence
\[
  \|\widehat{\diff f}(x)\|\leq \sqrt{2}(1+\sqrt{d-1})\,
  \mathrm{dist}(x,V_\bbC(g_f)).
\]
The bound now follows from
Proposition~\ref{prop:fundamentalproposition_aff},
together with $2^{3(n-1)}d+\sqrt{n}\leq 2^{3n-2}d$ 
and
\begin{align*}
\min\left\{\frac{2^{n-1}d}
{\ln{\left(1+2^{2-2n}\right)}}
+\frac{\sqrt{n}}{2},\frac{2^{2n}(d-1)}
{\ln{\left(1+2^{2-4n}\right)}}\right.
&\left.+\sqrt{\frac{n}{2}}\right\}\\
&\leq 2^{3n-4}d+\frac{\sqrt{n}}{2},
\end{align*}
for which we use that 
\[
1/\ln{\left(1+2^{2-2n}\right)}\leq 2^{2n-3}
\text{ and }
1/\ln{\left(1+2^{2-4n}\right)}\leq 2^{4n-3}.
\]
\end{proof}

The main difference between $C(f,x)$ and 
$\kappa(f,x)$ is that $C(f,x)$ is a non-linear
quantity and is hard to compute, while the local
condition number $\kappa(f,x)$---as indicated in Corollary \ref{cor:orthogonalprojection}---is a 
linear quantity and is rather easy to compute. 
Theorem~\ref{thm:MAIN1} below and 
the complexity analysis in Section~\ref{sec:probability} show that the local  
condition number $\kappaff(f,x)$ is easily 
amenable to the adaptive complexity analysis 
techniques developed
by Burr, Krahmer and Yap~\cite{burr2009,burr2016}. 

\subsection{Condition-based complexity}

The following result expresses a local size bound 
in terms of the local condition number 
$\kappaff(f,x)$ directly, without using the
construction in~\cite{burr2017}.

\begin{theo}\label{thm:MAIN1}
Assume that the interval approximation 
satisfies~\eqref{eq:intcond1} 
and~\eqref{eq:intcond2}. Then
\[
  x\mapsto 1/\left(2^{5/2}dn\kappaff(f,x)\right)^n
\]
is a local size bound  for $f$.
\end{theo}

\begin{proof}
Let $x\in \bbR^n$. As, by Theorem~\ref{theo:condprime}, $C_f'(J)$ implies $C_f(J)$, it is enough to compute the minimum volume of $J\in \mcI_n$ containing $x$ such that $C_f'(J)$ is false. This will still give a local size function for $f$.

Since $x\in J$, $\|x-m(J)\|\leq \sqrt{n}w(J)/2$. Hence, by Corollary~\ref{cor:lipschitz} and Proposition~\ref{prop:fundamentalproposition_aff}, either
\[
\big|\widehat{f}(m(J))\big|\geq \frac{1}{2\sqrt{2d}\,\kappaff(f,x)}
-(1+\sqrt{d})\sqrt{n}\,w(J)/2
\]
or
\[
\big|\widehat{\diff f}(m(J))\big|\geq \frac{1}{2\sqrt{2d}\,\kappaff(f,x)}-(1+\sqrt{d-1})
\sqrt{n}\,w(J)/2.\]
This means that $C'_f(J)$ is true if either
\[
  2\sqrt{2d}\,(1+\sqrt{d})\sqrt{n}\,\kappaff(f,x)w(J)< 1
\]
or
\[
2\sqrt{2d}\,(1+\sqrt{d-1})n\kappaff(f,x)w(J)<1.
\]
Hence we get that $C'_f(J)$ is true when both conditions are satisfied and the inequality
$1+\sqrt{d}\leq 2\sqrt{d}$ finishes the proof.
\end{proof}

Using the results above, we get the following 
theorem exhibiting a condition-based 
complexity analysis of
Algorithm~\ref{alg:PVAlgorithm}.

\begin{theo}\label{thm:MAIN2}
The number of $n$-cubes in the final subdivision 
of Algorithm~\ref{alg:PVAlgorithm} on input $(f,a)$ 
is at most
\[d^n\max\{1,a^n\}2^{n\log{n}+9n/2}\,\E_{x\in[-a,a]^n}\left(\kappaff(f,x)^n\right)
\]
if the interval approximation
satisfies~\eqref{eq:intcond1} 
and~\eqref{eq:intcond1}, and at most 
\[
  d^{2n}\max\{1,a^n\}2^{3n^2+2n}\,\E_{x\in[-a,a]^n}
  \left(\kappaff(f,x)^n\right)
\]
if the interval approximation satisfies the 
hypothesis of~\cite{burr2017}.
\end{theo}
\begin{proof}
This is just Theorems~\ref{theo:analysis2},~\ref{thm:MAIN1} and~\ref{thm:boundBGTbycond} combined with the fact that the integral $\int_{[-a,a]^n}\,\kappaff(f,x)^n\,\mathrm{d} x$ is nothing more than $(2a)^n\,\E_{x\in[-a,a]^n}\left(\kappaff(f,x)^n\right)$.
\end{proof}

We observe that in contrast with the 
complexity analyses (of condition numbers closely
related to $\kappaff$) in the  literature (see,
e.g.,~\cite{CKMW1,CKMW2,CKMW3,CKS16,BCL17,BCTC1}),
the bounds in Theorem~\ref{thm:MAIN2} depend on
$\E_{x\in[-a,a]^n}\left(\kappaff(f,x)^n\right)$ 
and not on $\max_{x\in[-a,a]}\kappaff(f,x)^n$. 
Whereas the former has finite expectation (over $f$), the latter has not. This shows that condition-based
analysis combined with adaptive complexity techniques 
such as continuous amortization 
may lead to substantial improvements.  
\section{Probabilistic analyses}\label{sec:probability}

In this section, we prove Theorems~\ref{expected} 
and~\ref{smoothed} stated in Section~\ref{sec:main}.

\subsection{Average Complexity Analysis}

The following theorem is the main technical result
from which the average complexity bound will follow. 

\begin{theo}\label{thm:boundlocalcondition}
Let $f\in\Pd$ be a dobro random polynomial with parameters $K$ and $\rho$. For all $x\in\bbR^n$ 
and $t\geq e^n$,
\[
\PP\left(\kappaff(f,x)^n \geq t\right)\leq 4 \left(\frac{c_1 c_2 K\rho \sqrt{N}}{\sqrt{n(n+1)}}\right)^{n+1} \frac{ \ln(t)^{\frac{n+1}{2}} }{t^{1+\frac{1}{n}}}
\]
where $c_1$ and $c_2$ are, respectively, the 
universal constants of Theorems~\ref{thm:probtool1}
and~\ref{thm:probtool2}.
\end{theo}

The proof of Theorem \ref{thm:boundlocalcondition}
relies on two basic results from geometric functional
analysis.

\begin{theo}\cite[Theorems~2.6.3 and~3.1.1]{V} \label{thm:probtool1} 
There is a universal constant $c_1\geq 1$ with 
the following property. For all random vectors $X\:=(X_1,\ldots,X_N)^T$ with each $X_i$ centered and sub-Gaussian with $\psi_2$-norm $K$, and for all 
$t \geq c_1K\sqrt{N}$ the following inequality is satisfied
\begin{equation}\tag*{\qed}
 \mathbb{P} \left( \norm{X} \geq t\right)  \leq  \exp\left(1-t^2/(c_1K)^2\right) .
 \end{equation}
\end{theo}

\begin{defi}
The \emph{concentration function} of a random 
vector $X\in\bbR^k$ is the function
$
\mcL_X(\varepsilon) := \max_{u \in \mathbb{R}^k} 
\PP \left( \norm{X - u}  \leq \varepsilon \right).
$
\end{defi}

\begin{theo}\cite[Corollary~1.4]{RV-1}\label{thm:probtool2}
There is a universal constant $c_2 \geq 1$ with 
the following property. For every random vector
$X=(X_1,\ldots,X_N)^T$ with independent random variables $X_i$, and every $k$-dimensional 
linear subspace $S$ of $\R^N$ we have
\[
  \mcL_{P_k(X)}\left(\varepsilon \sqrt{k} \right)
  \leq \left(c_2 \max_{1\leq i\leq N}\mcL_{X_i}(\varepsilon)\right)^k , 
\]
where $P_k$ is the orthogonal projection onto $S$.
\eproof
\end{theo}

\begin{remark}
In~\cite{grigoris16} and references therein one 
can find information about the optimal value of 
the absolute constant $c_2$ in
Theorem~\ref{thm:probtool2}.
\end{remark}

\begin{remark}\label{remark:boundconstants}
We notice that for a dobro random polynomial $f$ 
with parameters, $K$ and $\rho$ we have 
$K \rho \geq \frac{1}{4}$~\cite[(1)]{EPR18}. Actually, the product $K\rho$ is invariant 
under scaling in the following sense; for $t>0$, $tf$ is again a dobro polynomial 
with parameters $tK$ and $\rho/t$. Hence, for the sake of simplicity and without loss 
of generality, we will  assume $c_1 c_2 K\rho \geq 1$. Moreover, since $\kappaff$ is scale invariant, 
we can assume, again without loss of generality, that $c_1 K \geq 1$.
\end{remark}

\begin{proof}[Proof of Theorem~\ref{thm:boundlocalcondition}]
$\kappaff(f,x)=\|f\|/\|P_xf\|$ by Corollary~\ref{cor:orthogonalprojection}. So, by an union bound, for all $u,s>0$,
\begin{equation}\label{eq:unionbound}
   \PP \left( \kappaff(f,x) \geq s \right)  \leq 
   \PP \left( \norm{f} \geq u \right) + \PP \left( \norm{P_x f} \leq u/s \right) .  
\end{equation}
By Theorems~\ref{thm:probtool1} and~\ref{thm:probtool2}, we have
\[ 
 \PP (\kappaff(f,x) \geq s) \leq 
 \exp(1-u^2/(c_1K)^2) + \left(\frac{u c_2\rho}
 {s\sqrt{n+1}}\right)^{n+1}.
\]
We set $u=c_1K\sqrt{N\ln(s)}$ and use $s^{-N} \leq s^{-(n+1)}$, so we get
\[ 
 \PP (\kappaff(f,x) \geq s) \leq 
 4 \left(\frac{c_1  c_2 K\rho\sqrt{N}}
 {\sqrt{n+1}}\right)^{n+1}
 \frac{\ln(s)^{\frac{n+1}{2}}}{s^{n+1}}. 
\]
By substituting $s=t^{\frac{1}{n}}$ we are done.
\end{proof}

Combining Theorem~\ref{thm:MAIN2} with the next theorem, we get the proof of Theorem~\ref{expected}.

\begin{theo}\label{thm:main3}
Let $f\in\Pd$ be a dobro random polynomial with parameters $K$ and $\rho$. Then
\[
\E_f\E_{x\in [-a,a]^n}\left(\kappaff(f,x)^n\right)
\leq  d^{\frac{n^2+n}{2}} 2^{\frac{n^2+3\log(n)+9}{2}}
(c_1c_2K\rho)^{n+1} 
\]
where  $c_1$ and $c_2$ are the universal constants of
Theorems~\ref{thm:probtool1} and~\ref{thm:probtool2}.
\end{theo}

\begin{proof}
By the Fubini-Tonelli theorem,
\[
 \E_f\E_{x\in [-a,a]^n}\left(\kappaff(f,x)^n\right)
 =\E_{x\in [-a,a]^n}\E_f\left(\kappaff(f,x)^n\right)
\]
so it is enough to have a uniform bound for
\[
\E_f\left(\kappaff(f,x)^n\right)=\int_1^\infty
\PP\left(\kappaff(f,x)^n\geq t\right)\,\mathrm{d} t.
\]
Now, by Theorem~\ref{thm:boundlocalcondition}, this 
is bounded by
\[
e^n+4 \left(\frac{c_1c_2K\rho
\sqrt{N}}{\sqrt{n(n+1)}}\right)^{n+1}
\int_1^\infty\,
\frac{\ln(t)^{\frac{n+1}{2}}}{t^{1+1/n}}  
\, \mathrm{d} t.
\]
After the change of variables $t=e^{ns}$ the 
integral becomes
\[
 n \int_0^\infty\, (ns)^{\frac{n+1}{2}} e^{-s} 
 \,\mathrm{d} s = n^{\frac{n+3}{2}}
 \Gamma\left(\frac{n+3}{2}\right),
\]
where $\Gamma$ is Euler's Gamma function. Using 
the Stirling estimates for it, we obtain 
\[ 
 \Gamma\left(\frac{n+3}{2}\right) \leq 
 \sqrt{2\pi} \left(\frac{n+3}{2e}\right)
 ^{\frac{n+2}{2}}\leq  4
 \left(\frac{n+3}{4}\right)^{\frac{n+2}{2}}
\]
and $N \leq (2d)^n$. Combining all these 
inequalities, we obtain the desired upper bound.
\end{proof}

\subsection{Smoothed Complexity Analysis}

The tools used for our average complexity 
analysis yield also a smoothed complexity 
analysis (see~\cite{ST:02} 
or~\cite[\S2.2.7]{Condition}). We provide this 
analysis following the lines of~\cite{EPR19}, 

The main idea of smoothed complexity is to have a
complexity measure interpolating between 
worst-case complexity and average-case complexity. 
More precisely, we are interested in the 
maximum---over $f\in\Pd$---of the average 
cost of Algorithm~\ref{alg:PVAlgorithm} 
with input
\begin{equation}\label{eq:perturbed}
    q_{\sigma}:=f+\sigma \|f\|g
\end{equation}
where $g\in\Pd$ is a dobro random polynomials with parameters $K$ and $\rho$ and $\sigma\in(0,\infty)$. Notice that the perturbation $\sigma\|f\|g$ of $f$ 
is proportional to both $\sigma$ and $\|f\|$.

\begin{lem}\label{lem:smooth}
Let $q_\sigma$ be as in~\eqref{eq:perturbed}. Then 
for $t > 1+\sigma\sqrt{N}$
\[ 
\PP \left( \norm{q_\sigma} \geq t \norm{f} \right) \leq \exp\left(1 -(t-1)^2/\left(\sigma c_1 K\right)^2\right)
\]
and for every $x\in\bbR^n$,
\[ 
\PP \left( \norm{P_xq_\sigma} \leq \varepsilon \right) \leq \left(c_2 \rho \varepsilon /\left(\sigma\|f\|\sqrt{n+1}\right)\right)^{n+1}
\]
where $P_x$ is as in Corollary~\ref{cor:orthogonalprojection}.
\eproof
\end{lem}

\begin{proof}
By the triangle inequality we have 
$\PP(\norm{q_{\sigma}} \geq t \norm{f}) 
\leq \PP(\norm{g} \geq (t-1) / \sigma)$. 
Then we apply Theorem~\ref{thm:probtool1} 
which finishes the proof of the first claim.
The second claim is a direct consequence of 
Theorem~\ref{thm:probtool2}. 
\end{proof}

As in the average case, this leads to a tail bound.

\begin{theo} \label{thm:tailboundsmooth}
Let $q_\sigma$ be as in~\eqref{eq:perturbed}. Then 
for $\sigma>0$ and $t\geq e^n$, $\PP\left(\kappaff(q_\sigma,x)^n \geq t\right)$ 
is bounded by
\[
4 \left(\frac{c_1 c_2 K\rho \sqrt{N}}{\sqrt{n(n+1)}}\right)^{n+1} \frac{ \ln(t)^{\frac{n+1}{2}} }{t^{1+\frac{1}{n}}}\left(1+\frac{1}{\sigma}\right)^{n+1}
\]
where $c_1$ and $c_2$ are, respectively, the 
universal constants of Theorems~\ref{thm:probtool1} and~\ref{thm:probtool2}.
\end{theo}

\begin{proof}
We proceed as in the proof of Theorem~\ref{thm:boundlocalcondition}, but with Lemma~\ref{lem:smooth} using $u=\|f\|(\sigma c_1K\sqrt{N\ln(t)}+1)$. This gives the desired 
bound arguing as in that proof after noticing that
\[
  u\leq \|f\|(1+\sigma) c_1K\sqrt{N\ln(t)}
\]
which holds since $c_1K\sqrt{N\ln(t)}\geq 1$.
\end{proof}

Finally, the following theorem, together with Theorem~\ref{thm:MAIN2}, gives the proof of Theorem~\ref{smoothed}.

\begin{theo}
Let $q_\sigma$ be as in~\eqref{eq:perturbed}. 
Then for all $\sigma>0$, 
$\E_{q_\sigma}\E_{x\in [-a,a]^n}$ is bounded 
by
\begin{equation*}
d^{\frac{n^2+n}{2}} 2^{\frac{n^2+3\log(n)+9}{2}}
(c_1c_2K\rho)^{n+1}
\left(1+\frac{1}{\sigma}\right)^{n+1} 
\end{equation*}
where $c_1$ and $c_2$ are the universal constants 
of Theorems~\ref{thm:probtool1}
and~\ref{thm:probtool2}.
\end{theo}

\begin{proof}
The proof is as that of Theorem~\ref{thm:main3}, but
using Theorem~\ref{thm:tailboundsmooth} instead of
Theorem~\ref{thm:boundlocalcondition}. Actually, 
the integrand one ends up with is the same, up to 
a constant, so the calculation of the integral 
is the same up to that constant.
\end{proof}

\bibliographystyle{ACM-Reference-Format}
\bibliography{biblio.bib}


\begin{thebibliography}{00}


\ifx \showCODEN    \undefined \def \showCODEN     #1{\unskip}     \fi
\ifx \showDOI      \undefined \def \showDOI       #1{#1}\fi
\ifx \showISBNx    \undefined \def \showISBNx     #1{\unskip}     \fi
\ifx \showISBNxiii \undefined \def \showISBNxiii  #1{\unskip}     \fi
\ifx \showISSN     \undefined \def \showISSN      #1{\unskip}     \fi
\ifx \showLCCN     \undefined \def \showLCCN      #1{\unskip}     \fi
\ifx \shownote     \undefined \def \shownote      #1{#1}          \fi
\ifx \showarticletitle \undefined \def \showarticletitle #1{#1}   \fi
\ifx \showURL      \undefined \def \showURL       {\relax}        \fi
\providecommand\bibfield[2]{#2}
\providecommand\bibinfo[2]{#2}
\providecommand\natexlab[1]{#1}
\providecommand\showeprint[2][]{arXiv:#2}

\bibitem[\protect\citeauthoryear{B\"urgisser and Cucker}{B\"urgisser and
  Cucker}{2013}]%
        {Condition}
\bibfield{author}{\bibinfo{person}{Peter B\"urgisser} {and}
  \bibinfo{person}{Felipe Cucker}.} \bibinfo{year}{2013}\natexlab{}.
\newblock \bibinfo{booktitle}{{\em Condition}}. \bibinfo{series}{Grundlehren
  der mathematischen Wissenschaften}, Vol.~\bibinfo{volume}{349}.
\newblock \bibinfo{publisher}{Springer-Verlag}, \bibinfo{address}{Berlin}.
\newblock
\showDOI{%
\url{https://doi.org/10.1007/978-3-642-38896-5}}


\bibitem[\protect\citeauthoryear{B\"{u}rgisser, Cucker, and
  Lairez}{B\"{u}rgisser et~al\mbox{.}}{2018}]%
        {BCL17}
\bibfield{author}{\bibinfo{person}{Peter B\"{u}rgisser},
  \bibinfo{person}{Felipe Cucker}, {and} \bibinfo{person}{Pierre Lairez}.}
  \bibinfo{year}{2018}\natexlab{}.
\newblock \showarticletitle{Computing the Homology of Basic Semialgebraic Sets
  in Weak Exponential Time}.
\newblock \bibinfo{journal}{{\em J. ACM\/}} \bibinfo{volume}{66},
  \bibinfo{number}{1}, Article \bibinfo{articleno}{5} (\bibinfo{year}{2018}),
  \bibinfo{numpages}{30}~pages.
\newblock
\showISSN{0004-5411}
\showDOI{%
\url{https://doi.org/10.1145/3275242}}


\bibitem[\protect\citeauthoryear{B{\"{u}}rgisser, Cucker, and
  Tonelli{-}Cueto}{B{\"{u}}rgisser et~al\mbox{.}}{2018}]%
        {BCTC1}
\bibfield{author}{\bibinfo{person}{Peter B{\"{u}}rgisser},
  \bibinfo{person}{Felipe Cucker}, {and} \bibinfo{person}{Josu{\'{e}}
  Tonelli{-}Cueto}.} \bibinfo{year}{2018}\natexlab{}.
\newblock \showarticletitle{Computing the Homology of Semialgebraic Sets. {I:}
  Lax Formulas}.
\newblock  (\bibinfo{year}{2018}).
\newblock
\showeprint[arxiv]{1807.06435}
\newblock
\shownote{To appear at \JFoCM.}


\bibitem[\protect\citeauthoryear{Burr, Krahmer, and Yap}{Burr
  et~al\mbox{.}}{2009}]%
        {burr2009}
\bibfield{author}{\bibinfo{person}{Michael Burr}, \bibinfo{person}{Felix
  Krahmer}, {and} \bibinfo{person}{Chee Yap}.} \bibinfo{year}{2009}\natexlab{}.
\newblock \showarticletitle{Continuous amortization: A non-probabilistic
  adaptive analysis technique}.
\newblock \bibinfo{journal}{{\em Electronic Colloquium on Computational
  Complexity\/}} \bibinfo{number}{Report. No. 136} (\bibinfo{year}{2009}).
\newblock


\bibitem[\protect\citeauthoryear{Burr}{Burr}{2016}]%
        {burr2016}
\bibfield{author}{\bibinfo{person}{Michael~A. Burr}.}
  \bibinfo{year}{2016}\natexlab{}.
\newblock \showarticletitle{Continuous amortization and extensions: with
  applications to bisection-based root isolation}.
\newblock \bibinfo{journal}{{\em J. Symbolic Comput.\/}}  \bibinfo{volume}{77}
  (\bibinfo{year}{2016}), \bibinfo{pages}{78--126}.
\newblock
\showISSN{0747-7171}
\showDOI{%
\url{https://doi.org/10.1016/j.jsc.2016.01.007}}


\bibitem[\protect\citeauthoryear{Burr, Gao, and Tsigaridas}{Burr
  et~al\mbox{.}}{2017}]%
        {burr2017}
\bibfield{author}{\bibinfo{person}{Michael~A. Burr}, \bibinfo{person}{Shuhong
  Gao}, {and} \bibinfo{person}{Elias~P. Tsigaridas}.}
  \bibinfo{year}{2017}\natexlab{}.
\newblock \showarticletitle{The complexity of an adaptive subdivision method
  for approximating real curves}.
\newblock In \bibinfo{booktitle}{{\em I{SSAC}'17---{P}roceedings of the 2017
  {ACM} {I}nternational {S}ymposium on {S}ymbolic and {A}lgebraic
  {C}omputation}}. \bibinfo{publisher}{ACM}, \bibinfo{address}{New York},
  \bibinfo{pages}{61--68}.
\newblock
\showDOI{%
\url{https://doi.org/10.1145/3087604.3087654}}


\bibitem[\protect\citeauthoryear{Burr, Gao, and Tsigaridas}{Burr
  et~al\mbox{.}}{2018}]%
        {burr2018}
\bibfield{author}{\bibinfo{person}{Michael~A. Burr}, \bibinfo{person}{Shuhong
  Gao}, {and} \bibinfo{person}{Elias~P. Tsigaridas}.}
  \bibinfo{year}{2018}\natexlab{}.
\newblock \showarticletitle{The Complexity of Subdivision for Diameter-Distance
  Tests}.
\newblock  (\bibinfo{year}{2018}).
\newblock
\showeprint[arxiv]{1801.05864}


\bibitem[\protect\citeauthoryear{Cucker, Krick, Malajovich, and
  Wschebor}{Cucker et~al\mbox{.}}{2008}]%
        {CKMW1}
\bibfield{author}{\bibinfo{person}{Felipe Cucker}, \bibinfo{person}{Teresa
  Krick}, \bibinfo{person}{Gregorio Malajovich}, {and} \bibinfo{person}{Mario
  Wschebor}.} \bibinfo{year}{2008}\natexlab{}.
\newblock \showarticletitle{A numerical algorithm for zero counting. {I}:
  {C}omplexity and accuracy}.
\newblock \bibinfo{journal}{{\em \JoC\/}}  \bibinfo{volume}{24}
  (\bibinfo{year}{2008}), \bibinfo{pages}{582--605}.
\newblock
\showDOI{%
\url{https://doi.org/10.1016/j.jco.2008.03.001}}


\bibitem[\protect\citeauthoryear{Cucker, Krick, Malajovich, and
  Wschebor}{Cucker et~al\mbox{.}}{2009}]%
        {CKMW2}
\bibfield{author}{\bibinfo{person}{Felipe Cucker}, \bibinfo{person}{Teresa
  Krick}, \bibinfo{person}{Gregorio Malajovich}, {and} \bibinfo{person}{Mario
  Wschebor}.} \bibinfo{year}{2009}\natexlab{}.
\newblock \showarticletitle{A numerical algorithm for zero counting. {II}:
  {D}istance to ill-posedness and smoothed analysis}.
\newblock \bibinfo{journal}{{\em J. Fixed Point Theory Appl.\/}}
  \bibinfo{volume}{6} (\bibinfo{year}{2009}), \bibinfo{pages}{285--294}.
\newblock
\showDOI{%
\url{https://doi.org/10.1007/s11784-009-0127-4}}


\bibitem[\protect\citeauthoryear{Cucker, Krick, Malajovich, and
  Wschebor}{Cucker et~al\mbox{.}}{2012}]%
        {CKMW3}
\bibfield{author}{\bibinfo{person}{Felipe Cucker}, \bibinfo{person}{Teresa
  Krick}, \bibinfo{person}{Gregorio Malajovich}, {and} \bibinfo{person}{Mario
  Wschebor}.} \bibinfo{year}{2012}\natexlab{}.
\newblock \showarticletitle{A numerical algorithm for zero counting. {III}:
  {R}andomization and condition}.
\newblock \bibinfo{journal}{{\em Adv. Applied Math.\/}}  \bibinfo{volume}{48}
  (\bibinfo{year}{2012}), \bibinfo{pages}{215--248}.
\newblock
\showDOI{%
\url{https://doi.org/10.1016/j.aam.2011.07.001}}


\bibitem[\protect\citeauthoryear{{Cucker}, {Krick}, and {Shub}}{{Cucker}
  et~al\mbox{.}}{2018}]%
        {CKS16}
\bibfield{author}{\bibinfo{person}{Felipe {Cucker}}, \bibinfo{person}{Teresa
  {Krick}}, {and} \bibinfo{person}{Michael {Shub}}.}
  \bibinfo{year}{2018}\natexlab{}.
\newblock \showarticletitle{{Computing the Homology of Real Projective Sets}}.
\newblock \bibinfo{journal}{{\em Found. Comput. Math.\/}}  \bibinfo{volume}{18}
  (\bibinfo{year}{2018}), \bibinfo{pages}{929--970}.
\newblock
\showDOI{%
\url{https://doi.org/10.1007/s10208-017-9358-8}}


\bibitem[\protect\citeauthoryear{Demmel}{Demmel}{1988}]%
        {Demmel88}
\bibfield{author}{\bibinfo{person}{James Demmel}.}
  \bibinfo{year}{1988}\natexlab{}.
\newblock \showarticletitle{The probability that a numerical analysis problem
  is difficult}.
\newblock \bibinfo{journal}{{\em Math. Comp.\/}}  \bibinfo{volume}{50}
  (\bibinfo{year}{1988}), \bibinfo{pages}{449--480}.
\newblock
\showDOI{%
\url{https://doi.org/10.2307/2008617}}


\bibitem[\protect\citeauthoryear{Erg{\"u}r, Paouris, and Rojas}{Erg{\"u}r
  et~al\mbox{.}}{2018a}]%
        {EPR18}
\bibfield{author}{\bibinfo{person}{Alperen~A. Erg{\"u}r},
  \bibinfo{person}{Grigoris Paouris}, {and} \bibinfo{person}{J.~Maurice
  Rojas}.} \bibinfo{year}{2018}\natexlab{a}.
\newblock \showarticletitle{Probabilistic Condition Number Estimates for Real
  Polynomial Systems I: A Broader Family of Distributions}.
\newblock \bibinfo{journal}{{\em Found. Comput. Math.\/}}
  (\bibinfo{year}{2018}).
\newblock
\showISSN{1615-3383}
\showDOI{%
\url{https://doi.org/10.1007/s10208-018-9380-5}}


\bibitem[\protect\citeauthoryear{Erg{\"u}r, Paouris, and Rojas}{Erg{\"u}r
  et~al\mbox{.}}{2018b}]%
        {EPR19}
\bibfield{author}{\bibinfo{person}{Alperen~A. Erg{\"u}r},
  \bibinfo{person}{Grigoris Paouris}, {and} \bibinfo{person}{J.~Maurice
  Rojas}.} \bibinfo{year}{2018}\natexlab{b}.
\newblock \showarticletitle{{Probabilistic Condition Number Estimates for Real
  Polynomial Systems II: Structure and Smoothed Analysis}}.
\newblock  (\bibinfo{year}{2018}).
\newblock
\showeprint[arxiv]{1809.03626}


\bibitem[\protect\citeauthoryear{Livshyts, Paouris, and Pivovarov}{Livshyts
  et~al\mbox{.}}{2016}]%
        {grigoris16}
\bibfield{author}{\bibinfo{person}{Galyna Livshyts}, \bibinfo{person}{Grigoris
  Paouris}, {and} \bibinfo{person}{Peter Pivovarov}.}
  \bibinfo{year}{2016}\natexlab{}.
\newblock \showarticletitle{On sharp bounds for marginal densities of product
  measures}.
\newblock \bibinfo{journal}{{\em Israel Journal of Mathematics\/}}
  \bibinfo{volume}{216}, \bibinfo{number}{2} (\bibinfo{year}{2016}),
  \bibinfo{pages}{877–889}.
\newblock
\showISSN{1565-8511}
\showDOI{%
\url{https://doi.org/10.1007/s11856-016-1431-5}}


\bibitem[\protect\citeauthoryear{Plantinga and Vegter}{Plantinga and
  Vegter}{2004}]%
        {plantingavegter2004}
\bibfield{author}{\bibinfo{person}{Simon Plantinga} {and} \bibinfo{person}{Gert
  Vegter}.} \bibinfo{year}{2004}\natexlab{}.
\newblock \showarticletitle{Isotopic Approximation of Implicit Curves and
  Surfaces}. In \bibinfo{booktitle}{{\em Proceedings of the 2004
  Eurographics/ACM SIGGRAPH Symposium on Geometry Processing}} {\em
  (\bibinfo{series}{SGP '04})}. \bibinfo{publisher}{ACM}, \bibinfo{address}{New
  York, NY, USA}, \bibinfo{pages}{245--254}.
\newblock
\showISBNx{3-905673-13-4}
\showDOI{%
\url{https://doi.org/10.1145/1057432.1057465}}


\bibitem[\protect\citeauthoryear{Ratschek and Rokne}{Ratschek and
  Rokne}{1984}]%
        {ratschek1984}
\bibfield{author}{\bibinfo{person}{Helmut Ratschek} {and} \bibinfo{person}{Jon
  Rokne}.} \bibinfo{year}{1984}\natexlab{}.
\newblock \bibinfo{booktitle}{{\em Computer methods for the range of
  functions}}.
\newblock \bibinfo{publisher}{Ellis Horwood Ltd., Chichester; Halsted Press
  [John Wiley \& Sons, Inc.], New York}. 168 pages.
\newblock
\showISBNx{0-85312-703-4}


\bibitem[\protect\citeauthoryear{Renegar}{Renegar}{1995}]%
        {Renegar95}
\bibfield{author}{\bibinfo{person}{James Renegar}.}
  \bibinfo{year}{1995}\natexlab{}.
\newblock \showarticletitle{Incorporating condition measures into the
  complexity theory of linear programming}.
\newblock \bibinfo{journal}{{\em \SIOPT\/}}  \bibinfo{volume}{5}
  (\bibinfo{year}{1995}), \bibinfo{pages}{506--524}.
\newblock
\showDOI{%
\url{https://doi.org/10.1137/0805026}}


\bibitem[\protect\citeauthoryear{Rudelson and Vershynin}{Rudelson and
  Vershynin}{2008}]%
        {RV}
\bibfield{author}{\bibinfo{person}{Mark Rudelson} {and} \bibinfo{person}{Roman
  Vershynin}.} \bibinfo{year}{2008}\natexlab{}.
\newblock \showarticletitle{The {L}ittlewood-{O}fford problem and invertibility
  of random matrices}.
\newblock \bibinfo{journal}{{\em Adv. Math.\/}} \bibinfo{volume}{218},
  \bibinfo{number}{2} (\bibinfo{year}{2008}), \bibinfo{pages}{600--633}.
\newblock
\showISSN{0001-8708}
\showDOI{%
\url{https://doi.org/10.1016/j.aim.2008.01.010}}


\bibitem[\protect\citeauthoryear{Rudelson and Vershynin}{Rudelson and
  Vershynin}{2015}]%
        {RV-1}
\bibfield{author}{\bibinfo{person}{Mark Rudelson} {and} \bibinfo{person}{Roman
  Vershynin}.} \bibinfo{year}{2015}\natexlab{}.
\newblock \showarticletitle{Small ball probabilities for linear images of
  high-dimensional distributions}.
\newblock \bibinfo{journal}{{\em Int. Math. Res. Not. IMRN\/}}
  \bibinfo{number}{19} (\bibinfo{year}{2015}), \bibinfo{pages}{9594--9617}.
\newblock
\showISSN{1073-7928}
\showDOI{%
\url{https://doi.org/10.1093/imrn/rnu243}}


\bibitem[\protect\citeauthoryear{Spielman and Teng}{Spielman and Teng}{2002}]%
        {ST:02}
\bibfield{author}{\bibinfo{person}{Daniel~A. Spielman} {and}
  \bibinfo{person}{Shang-Hua Teng}.} \bibinfo{year}{2002}\natexlab{}.
\newblock \showarticletitle{Smoothed Analysis of Algorithms}. In
  \bibinfo{booktitle}{{\em Proceedings of the International Congress of
  Mathematicians}}, Vol.~\bibinfo{volume}{I}. \bibinfo{pages}{597--606}.
\newblock


\bibitem[\protect\citeauthoryear{Vershynin}{Vershynin}{2018}]%
        {V}
\bibfield{author}{\bibinfo{person}{Roman Vershynin}.}
  \bibinfo{year}{2018}\natexlab{}.
\newblock \bibinfo{booktitle}{{\em High-dimensional probability: An
  introduction with applications in data science}}. \bibinfo{series}{Cambridge
  Series in Statistical and Probabilistic Mathematics},
  Vol.~\bibinfo{volume}{47}.
\newblock \bibinfo{publisher}{Cambridge University Press, Cambridge}.
\newblock
\showISBNx{978-1-108-41519-4}
\showDOI{%
\url{https://doi.org/10.1017/9781108231596}}


\end{thebibliography}
\end{document}